\documentclass[11pt,a4paper]{article}
\usepackage[T1]{fontenc}
\usepackage{lmodern}
\usepackage{microtype}
\usepackage[margin=27mm]{geometry}
\usepackage{amsmath,amssymb,amsthm}
\usepackage{graphicx}
\usepackage{enumitem}
\usepackage[font=small,labelfont=bf,labelsep=period]{caption}
\usepackage{hyperref}
\usepackage{orcidlink}
\hypersetup{hidelinks,
 pdftitle={Exact Reachability under Vertex- and Area-Centroid Control},
 pdfauthor={Yongjie Guan},
 pdfsubject={Computational geometry, geometric control, and exact reachability},
 pdfkeywords={vertex-centroid moves, area-centroid flows, affine-equivariant centers, exact reachability, one-sided shape control}}
\graphicspath{{figures/}}
\newtheorem{theorem}{Theorem}
\newtheorem{proposition}[theorem]{Proposition}
\newtheorem{lemma}[theorem]{Lemma}
\newtheorem{corollary}[theorem]{Corollary}
\theoremstyle{remark}

\newcommand{\R}{\mathbb R}
\newcommand{\one}{\mathbf 1}
\newcommand{\cM}{\mathcal M}
\newcommand{\cS}{\mathcal S}
\newcommand{\fg}{\mathfrak g}
\newcommand{\X}{\mathcal X}
\DeclareMathOperator{\rank}{rank}
\DeclareMathOperator{\Lie}{Lie}
\DeclareMathOperator{\tr}{tr}
\DeclareMathOperator{\Sym}{Sym}
\DeclareMathOperator{\sym}{sym}
\title{Exact Reachability under Vertex- and Area-Centroid Control}
\author{Yongjie Guan\,\orcidlink{0009-0002-1071-2741}\\
{\normalsize Zhejiang University of Technology}}
\date{}
\begin{document}
\maketitle

\begin{abstract}
A vertex-centroid move slides one vertex of a polygon along the line through
it and the centroid of all vertices; Problem~60 of The Open Problems Project
asks whether such moves can make every polygon regular.  We prove that two
simple labeled $n$-gons, with straight angles allowed, are joined through
simple polygons by finitely many positive moves exactly when they have the
same orientation; a nonconstant move is positive if the vertex remains on the
same side of the centroid of the other vertices.  The same holds for flows
toward or away from the instantaneous area centroid.  Inward-only or
outward-only moves reach every target of the same orientation up to
translation and positive scaling.

The common mechanism is a first-order criterion: a symmetric definite center
response makes the one-vertex fields of a smooth affine-equivariant center
bracket generating.  The response is a weighted covariance for point
barycenters and $(d+2)$ times the volume covariance for polyhedral
centroids.  Once translation and scale are factored out, the fields sum to
zero, giving exact one-sided control.  We resolve Problem~60 under both
conventions for a vertex at the centroid.  For the vertex-centroid
system on $n$ affinely spanning points in $\R^d$, we classify positive-move
reachability and prove a local $3dn$-move bound.
\end{abstract}

\noindent\textbf{Keywords.} Exact reachability; centroid control;
affine-equivariant centers; one-sided shape control; polygon morphing.

\section{Introduction}\label{sec:introduction}

Problem~60 of The Open Problems Project (TOPP)~\cite{TOPP60} asks whether
every polygon can be made regular by finitely many \emph{vertex-centroid
moves}, each sliding one vertex along the line through it and the centroid
of all vertices; the entry also proposes an area-centroid variant, in which
the active vertex pursues the instantaneous area centroid.  Both systems
steer a simple polygon exactly to any target of the same orientation while
preserving simplicity.  Inward-only or outward-only moves suffice to reach
every such target shape.

\paragraph{Model.}
Let $n\ge2$, and regard a labeled configuration $P=(p_1,\ldots,p_n)$ of
points in $\R^d$ as a $d\times n$ matrix with arithmetic centroid
$m(P)=n^{-1}\sum_jp_j$.  The centroid $c_i(P)=(n-1)^{-1}\sum_{j\ne i}p_j$ of
the other vertices stays fixed while $p_i$ moves, and
$m(P)=c_i(P)+n^{-1}(p_i-c_i(P))$.  Hence the line through $p_i$ and $m(P)$,
when defined, passes through $c_i(P)$, and a vertex-centroid move has
endpoint
\begin{equation}\label{eq:move}
 p_i'=c_i(P)+\lambda\bigl(p_i-c_i(P)\bigr),\qquad p_j'=p_j\quad(j\ne i),
\end{equation}
for some $\lambda\in\R$; it is executed along the straight segment from
$p_i$ to $p_i'$.  The move is \emph{positive} if $\lambda>0$, \emph{inward} if
$0<\lambda<1$, and \emph{outward} if $\lambda>1$.  A nonconstant positive
move keeps $p_i$ on the same open ray from $c_i(P)$ and is undone by the move
with parameter $\lambda^{-1}$.  When $p_i=m(P)$ at the start of a move,
equivalently $p_i=c_i(P)$, the line is undefined and~\eqref{eq:move} gives
the identity.  This is the default \emph{immobile convention}; the
\emph{permissive convention} lets such a vertex depart along any line.

A \emph{simple polygon} is a straight-line embedding in the plane of the
cycle $p_1p_2\cdots p_np_1$, where $n\ge3$ and indices are cyclic; vertices
with straight angles are allowed.  Its orientation is the sign of its signed
area $\mathcal A(P)$, and its area centroid is $g(P)$:
\[
 \mathcal A(P)=\frac12\sum_i\det(p_i,p_{i+1}),\qquad
 g(P)=\frac{1}{6\mathcal A(P)}\sum_i(p_i+p_{i+1})\det(p_i,p_{i+1}).
\]
An \emph{area-centroid move} is a flow segment of
\[
 \dot p_i=p_i-g(P),\qquad \dot p_j=0\quad(j\ne i),
\]
run forward (outward) or backward (inward).  Because $g(P)$ moves with
$p_i$, the path of the vertex can curve (Figure~\ref{fig:area}).

An \emph{execution} is a finite sequence of moves; flow segments run over
compact subintervals of their existence intervals, and identity moves are
not counted.  An execution stays in a set if every intermediate
configuration does.  Finally, $P$ and $Q$ have the same \emph{shape} if
$P=aQ+b\one^{\mathsf T}$ for some $a>0$ and $b\in\R^d$; rotations are not
factored out.

\begin{figure}[!htbp]
\centering
\includegraphics[width=.88\linewidth]{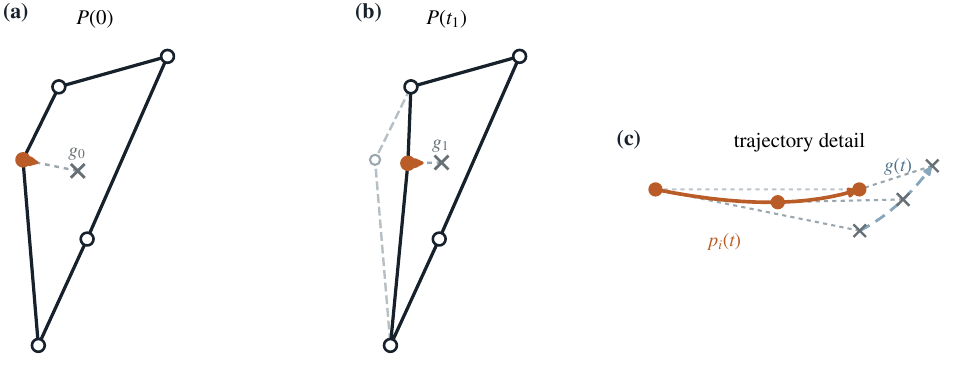}
\caption{An area-centroid flow can curve.  Panels (a) and (b) show an inward
segment: the active vertex (orange) and the area centroid move together, and
the line joining them gives the velocity direction at each time.  Panel (c)
traces both paths; arrows indicate the direction of the inward execution.
An orientation-preserving affine transformation magnifies the curvature for
display.}
\label{fig:area}
\end{figure}

\paragraph{Results.}
For both centroids, orientation is the only obstruction to exact control
between simple polygons and to shape control under one-sided motion.

\begin{theorem}[Exact polygon control]\label{thm:polygons}
Let $P$ and $Q$ be simple labeled $n$-gons.  For the vertex-centroid system
and for the area-centroid system, the following are equivalent:
\begin{enumerate}[label=\textup{(\roman*)},nosep]
\item $P$ and $Q$ have the same orientation.
\item A finite execution through simple polygons takes $P$ exactly to $Q$;
      vertex-centroid moves are positive, and area-centroid flows may run in
      either direction.
\item For each direction, inward or outward, a finite execution through
      simple polygons using only moves in that direction takes $P$ to
      $aQ+b\one^{\mathsf T}$ for some $a>0$ and $b\in\R^2$.
\end{enumerate}
\end{theorem}

The scale in~\textup{(iii)} cannot be prescribed in general
(Section~\ref{sec:one-sided}).
Straight angles are essential: along any path of simple polygons without
straight angles, each vertex retains its convex or reflex type.

We formulate Problem~60 for arbitrary labeled planar configurations,
allowing repeated points and self-intersections, with no constraint during
the motion.  The target is a nondegenerate regular $n$-gon with
$p_1,\ldots,p_n$ in cyclic order.

\begin{corollary}[TOPP Problem~60]\label{cor:topp}
For vertex-centroid moves with $n\ge3$, a labeled planar configuration of
$n$ points can be made regular under the immobile convention if and only if
it is not collinear; in that case, inward-only moves suffice, as do
outward-only moves.  Under the
permissive convention, every labeled planar configuration of $n$ points can
be made regular.
\end{corollary}

\paragraph{Approach.}
The proofs separate control from topology.  A symmetric definite first-order
response of an affine-equivariant center gives
full Lie rank for its one-vertex fields (Section~\ref{sec:response}).  Covariance identities
verify this condition for barycenters and for area and volume centroids
(Section~\ref{sec:covariance}).  Factoring out translation and scale makes
the fields sum to zero, which gives finite one-sided reachability
(Section~\ref{sec:one-sided}).  Connected components then yield the
classifications (Section~\ref{sec:topology}); in particular,
Theorem~\ref{thm:configurations} gives reversible control in every dimension
under any connected open constraint on affinely spanning configurations.
For the arithmetic centroid, explicit matrix words give a local move bound
and the exact move group (Section~\ref{sec:counts} and
Appendix~\ref{app:words}).

\paragraph{Related work.}
Hawley's two-step placement procedure~\cite[pp.~4--9]{Hawley2014} requires
the reference vertex-centroid line to meet the line joining the moving vertex
to its target.  He proposes remedies for degenerate placements but does not
establish their finite termination; the final move uses the permissive
convention.  Our results hold under the immobile convention, keep polygons
simple, and cover the area centroid and higher dimensions.

Planar morphing~\cite{Cairns1944,Thomassen1983,AlamdariEtAl2017} and single-vertex
convexification~\cite{AbregoEtAl2012} let vertices move in arbitrary
directions, and formation control~\cite{ChenEtAl2017} has independent edge
inputs; our system uses a single scalar radial control for the active
vertex.  Affine-invariant points~\cite{MeyerSchuettWerner2015}
provide a related framework for equivariant centers of convex bodies, and
projective centroid
variations~\cite{Izmestiev2023} underlie the volume covariance identity.
The orbit and accessibility theorems of Sussmann and
Krener~\cite{Sussmann1973,Krener1974} supply the control tools.

\section{A response criterion for exact control}\label{sec:response}

Let $d\ge1$ and $n\ge d+1$, and let
$\mathcal U\subseteq(\R^d)^n$ be open and invariant under
orientation-preserving affine maps.  A smooth map $c:\mathcal U\to\R^d$ is
an \emph{affine-equivariant center} if
\begin{equation}\label{eq:equivariance}
 c(AP+b\one^{\mathsf T})=Ac(P)+b\qquad(A\in\mathrm{GL}^+(d,\R),\ b\in\R^d).
\end{equation}
Let $\iota_i(v)$ be the configuration velocity that moves vertex $i$ with
velocity $v$ and fixes the others, and define
\begin{equation}\label{eq:response}
 q_i=p_i-c(P),\qquad \X_i(P)=\iota_i(q_i),\qquad
 d_i=Dc_P[\X_i(P)],\qquad
 \mathcal R_c(P)=\sum_i d_iq_i^{\mathsf T}.
\end{equation}
Thus $d_i$ is the velocity of the center when vertex $i$ moves radially
away from it, and the \emph{response} $\mathcal R_c$ collects these
velocities.  A \emph{center move} is a flow segment of some $\X_i$ in either
time direction; since $\X_i$ vanishes where $q_i=0$, a vertex at the center
stays there, as in the immobile convention.  Under the affine map
in~\eqref{eq:equivariance},
\begin{equation}\label{eq:response-affine}
 q_i\mapsto Aq_i,\qquad d_i\mapsto Ad_i,\qquad
 \mathcal R_c\mapsto A\mathcal R_cA^{\mathsf T},
\end{equation}
so symmetry and definiteness of the response are affine invariant.  We use
the bracket convention ${[X,Y]=DY\,X-DX\,Y}$.

\begin{theorem}[Definite-response criterion]\label{thm:response}
If $\mathcal R_c(P)$ is symmetric and definite at $P\in\mathcal U$, then
\[
 \Lie(\X_1,\ldots,\X_n)(P)=T_P\mathcal U.
\]
Consequently, if $\mathcal R_c$ is symmetric and definite on a connected open
set $\Omega\subseteq\mathcal U$, then any two points of $\Omega$ are joined by
finitely many center moves, each contained in $\Omega$.
\end{theorem}

\begin{proof}
Fix $P$ and let $E=\Lie(\X_1,\ldots,\X_n)(P)$.  Affine maps preserve the
fields, and hence $\dim E$, so by~\eqref{eq:response-affine} we may assume
that $c(P)=0$ and $\mathcal R_c(P)=\varepsilon I_d$ with
$\varepsilon\in\{1,-1\}$.  With $Q_c=(q_1\ \cdots\ q_n)$ and
$D=(d_1\ \cdots\ d_n)$, this reads
\begin{equation}\label{eq:normalization}
 DQ_c^{\mathsf T}=Q_cD^{\mathsf T}=\varepsilon I_d.
\end{equation}
In particular, $D$ is surjective.

Represent a covector annihilating $E$ by $L=(\ell_1\ \cdots\ \ell_n)$, so
that it maps a velocity $V$ to $\sum_i\ell_i^{\mathsf T}v_i$.  Since
$[\X_i,\X_j]=\iota_i(d_j)-\iota_j(d_i)$ for $i\ne j$, annihilating the fields
and their first brackets gives
\begin{equation}\label{eq:conormal}
 \ell_i^{\mathsf T}q_i=0\quad\text{for every }i,\qquad
 L^{\mathsf T}D=D^{\mathsf T}L.
\end{equation}
If $Dv=0$, the second identity gives $D^{\mathsf T}Lv=0$, so $Lv=0$ by
surjectivity of $D$.  Hence $L=HD$ for a unique $d\times d$ matrix $H$, and
the second identity in~\eqref{eq:conormal} becomes
$D^{\mathsf T}(H-H^{\mathsf T})D=0$, so $H$ is symmetric.  Summing the first
identities gives $0=\sum_id_i^{\mathsf T}Hq_i=\tr(HQ_cD^{\mathsf T})
=\varepsilon\tr H$.  Thus $L\mapsto H$ identifies $E^\perp$ with a subspace
$\mathcal H$ of the traceless symmetric matrices $\Sym_0(d)$ satisfying
\begin{equation}\label{eq:pointwise-H}
 d_i^{\mathsf T}Hq_i=0\qquad\text{for every }H\in\mathcal H\text{ and every }i.
\end{equation}

We show that $\mathcal H=0$, beginning with its rotation invariance.  For
$B\in M_d(\R)$, let $Z_B(R)=BR$.  By equivariance, $[Z_B,\X_i]=0$, so $Z_B$
commutes with every field in $\Lie(\X_1,\ldots,\X_n)$.  By~\eqref{eq:normalization}, the
covector $HD$ takes the value $\sum_id_i^{\mathsf T}HBq_i=\varepsilon\tr(HB)$
on $Z_B(P)$, so
\[
 \mathcal K:=\{B\in M_d(\R):Z_B(P)\in E\}=\mathcal H^\perp
\]
for the trace pairing $\langle H,B\rangle=\tr(HB)$.  The space $\mathcal K$
is a Lie algebra.  Indeed, for $A,B\in\mathcal K$ choose
$U,V\in\Lie(\X_1,\ldots,\X_n)$ with $U(P)=Z_A(P)$ and $V(P)=Z_B(P)$.  Since
$Z_A-U$ and $Z_B-V$ vanish at $P$, so does their bracket, and since $Z_A,Z_B$
commute with $U,V$,
\[
 0=[Z_A-U,Z_B-V](P)=[Z_A,Z_B](P)+[U,V](P).
\]
As $[Z_A,Z_B]=Z_{BA-AB}$ and $[U,V](P)\in E$, we get $BA-AB\in\mathcal K$.

Every skew-symmetric $A$ lies in $\mathcal K$, because $\tr(HA)=0$ for
symmetric $H$.  For such $A$ and for $H\in\mathcal H$, $B\in\mathcal K$,
\[
 \tr((AH-HA)B)=\tr(H(BA-AB))=0,
\]
so $AH-HA\in\mathcal K^\perp=\mathcal H$, and $\mathcal H$ is invariant under
conjugation by $\mathrm{SO}(d)$.  For $d\ge2$ this action on $\Sym_0(d)$ is
irreducible.  Indeed, rotate a nonzero element of an invariant subspace to
$H=\operatorname{diag}(h_1,\ldots,h_d)$ and choose $h_a\ne h_b$.  With matrix
units $E_{ab}$, put $J_{ab}=E_{ab}-E_{ba}$ and $S_{ab}=E_{ab}+E_{ba}$.  Then
\[
 J_{ab}H-HJ_{ab}=(h_b-h_a)S_{ab},\qquad
 J_{ab}S_{ab}-S_{ab}J_{ab}=2(E_{aa}-E_{bb}).
\]
Rotations send $S_{ab}$ to every off-diagonal symmetric basis matrix up to
sign, and the second identity gives all diagonal differences.  Together
these span $\Sym_0(d)$.  Thus $\mathcal H=0$ or $\mathcal H=\Sym_0(d)$;
for $d=1$, $\Sym_0(1)$ is already zero.

If $\mathcal H=\Sym_0(d)$ with $d\ge2$, then~\eqref{eq:pointwise-H} makes
every $\sym(q_id_i^{\mathsf T})$ a multiple of $I_d$.  For $q_i\ne0$, its
quadratic form vanishes on $q_i^\perp$, so $\sym(q_id_i^{\mathsf T})=0$;
multiplying by $q_i$ then forces $d_i=0$.
Every $q_id_i^{\mathsf T}$ therefore vanishes,
contradicting~\eqref{eq:normalization}.  Hence $\mathcal H=0$, that is,
$E=T_P\mathcal U$.

By Sussmann's orbit theorem~\cite{Sussmann1973}, full Lie rank makes every
orbit in $\Omega$ open.  These orbits consist of points joined by finitely
many flow segments inside $\Omega$.  They partition the connected set
$\Omega$, so there is only one.
\end{proof}

At a square, the arithmetic-centroid fields and their first brackets have
rank $7$; higher brackets supply the eighth direction.  In contrast, the
intersection of transverse diagonal
lines $p_1p_3$ and $p_2p_4$ of a quadrilateral is an affine-equivariant center
with zero response, since radial vertex motions preserve both lines.  Its
orbit dimension equals the number of vertices distinct from the center, at
most four.  Appendix~\ref{app:examples} gives these calculations and an
affine-equivariant center with nonsymmetric response.

\section{Covariance identities}\label{sec:covariance}

\paragraph{Point barycenters.}
For fixed weights $\omega_i>0$ with $\sum_i\omega_i=1$, the barycenter
$c_\omega(P)=\sum_i\omega_ip_i$ has $d_i=\omega_iq_i$, so its
response~\eqref{eq:response} is the weighted covariance
\begin{equation}\label{eq:discrete-covariance}
 \mathcal R_{c_\omega}(P)=\sum_i\omega_iq_iq_i^{\mathsf T}.
\end{equation}
It is positive definite exactly when the points affinely span $\R^d$, that
is, on
\[
 \cM_{d,n}=\bigl\{P:\rank\widehat P=d+1\bigr\},\qquad
 \widehat P=\begin{pmatrix}P\\\one^{\mathsf T}\end{pmatrix}.
\]
For uniform weights the center is $m(P)$, and
$p_i-m(P)=\frac{n-1}{n}(p_i-c_i(P))$.  Since $c_i(P)$ is constant along the
flow of $\X_i$,
\begin{equation}\label{eq:arithmetic-flow}
 p_i(t)=c_i(P)+e^{(n-1)t/n}\bigl(p_i-c_i(P)\bigr).
\end{equation}
The correspondence $\lambda=e^{(n-1)t/n}$ identifies finite flow segments
with positive moves~\eqref{eq:move}, preserving their paths and inward or
outward direction.

\paragraph{Area and volume centroids.}
Fix an oriented simplicial complex with vertices $1,\ldots,n$.  Let
$\mathcal U$ be an open set of configurations, invariant under
orientation-preserving affine maps, on which the complex is realized by
nondegenerate simplices as the embedded, outward-oriented boundary of a
bounded Lipschitz polyhedral domain $K(P)$.  Simple polygons are the planar
case, with the cycle reversed on the negatively oriented
component.  Let
\[
 g(P)=\frac1{|K(P)|}\int_{K(P)}x\,dx,\qquad
 \Sigma_{K(P)}=\frac1{|K(P)|}\int_{K(P)}(x-g(P))(x-g(P))^{\mathsf T}\,dx.
\]
Signed-simplex formulas for volume and first moment make $g$ smooth on
$\mathcal U$, and a change of variables makes it affine equivariant.

\begin{proposition}[Volume response]\label{prop:volume-response}
For the volume centroid,
\begin{equation}\label{eq:volume-response}
 \mathcal R_g(P)=(d+2)\Sigma_{K(P)}.
\end{equation}
In particular, the area-centroid response of a simple polygon is
$4\Sigma_{K(P)}$, and all these responses are positive definite.
\end{proposition}

\begin{proof}
Translate so that $g(P)=0$, whence $q_i=p_i$, and fix $y\in\R^d$.  For small
$|t|$, the projective map
\[
 \phi_t(x)=\frac{x}{1-t\,y^{\mathsf T}x}
\]
has positive denominator on a neighborhood of $\overline{K(P)}$ and is a
diffeomorphism there.  It maps each boundary simplex onto the simplex spanned
by the image vertices, so $K(\phi_t(P))=\phi_t(K(P))$.  At $t=0$ the vertex velocities are $(y^{\mathsf T}p_i)p_i$,
so this motion is $\sum_i(y^{\mathsf T}p_i)\X_i(P)$ and
\begin{equation}\label{eq:projective-response}
 \left.\frac{d}{dt}\right|_{t=0}g(\phi_t(P))
 =\sum_i(y^{\mathsf T}p_i)\,d_i=\mathcal R_g(P)y.
\end{equation}
On the other hand, the Jacobian determinant of $\phi_t$ is
$(1-t\,y^{\mathsf T}x)^{-(d+1)}$, so
\[
 g(\phi_t(P))=
 \frac{\displaystyle\int_{K(P)}x(1-t\,y^{\mathsf T}x)^{-(d+2)}\,dx}
      {\displaystyle\int_{K(P)}(1-t\,y^{\mathsf T}x)^{-(d+1)}\,dx}.
\]
The numerator vanishes at $t=0$, so the derivative there is
$(d+2)|K(P)|^{-1}\int_{K(P)}xx^{\mathsf T}\,dx\;y=(d+2)\Sigma_{K(P)}y$,
proving~\eqref{eq:volume-response}.  Finally, $\Sigma_{K(P)}$ is positive
definite because a nonzero linear functional cannot vanish identically on
the nonempty interior of $K(P)$.
\end{proof}

Equation~\eqref{eq:projective-response} identifies the projective first
variation with the one-vertex response, connecting the control system to the
inertia formulas of Izmestiev~\cite[Lemma~6.3]{Izmestiev2023}.  The covariance
identity also holds for nonconvex domains, including those whose centroid
lies outside $K(P)$.

\section{One-sided control of shape}\label{sec:one-sided}

Theorem~\ref{thm:response} uses moves in both time directions.  After
translation and scale are factored out, the one-vertex fields sum to zero.
The following lemma turns this balance into exact one-sided reachability.

\begin{lemma}[Zero-sum switching]\label{lem:zero-sum}
Let $Y_1,\ldots,Y_n$ be smooth bracket-generating vector fields on a
connected manifold $M$ with $\sum_iY_i=0$.  Then any two points of $M$ are
joined by finitely many flow segments of the $Y_i$ with positive times, and
likewise by finitely many with negative times.
\end{lemma}

\begin{proof}
Let $\mathcal A^+(x)$ and $\mathcal A^-(x)$ be the sets reachable from $x$ by
finite words with positive and with negative times, respectively.  By the
orbit theorem, words with times of both signs reach all of $M$.  Since
$-Y_i=\sum_{j\ne i}Y_j$, Lie--Trotter splitting approximates each negative
segment by positive words:
\[
 \varphi_i^{-t}=\lim_{N\to\infty}
 \left(\varphi_{j_{n-1}}^{t/N}\circ\cdots\circ
       \varphi_{j_1}^{t/N}\right)^N,\qquad t>0,
\]
where $j_1,\ldots,j_{n-1}$ list the indices other than $i$.  The convergence
is uniform on compact flow tubes, and partial cycles move points by
$O(N^{-1})$, so sufficiently fine words stay in
$M$~\cite[Lemmas~A.6--A.7]{BeauchardEtAl2026}.  Hence $\mathcal A^+(x)$ is
dense in $M$.

Bracket generation also gives $\mathcal A^-(y)$ nonempty interior; we recall
Krener's argument~\cite{Krener1974} in finite form.  Among endpoint maps of
negative words from $y$, choose one attaining the maximal differential rank
$r$ at strictly negative durations, and fix all but $r$ independent duration
parameters.  After shrinking the parameter neighborhood, all durations
remain negative and the restricted map $F$ is an embedding of rank $r$.
Its image $S$ lies in $\mathcal A^-(y)$.  For each $i$, the map
$G_i(u,s)=\varphi_i^s(F(u))$ is again a negative-word endpoint map for small
$s<0$, so $\rank DG_i\le r$ there and, by lower semicontinuity of rank, also
at $s=0$.  Since $D_uG_i(u,0)$ already spans $T_{F(u)}S$, the derivative
$\partial_sG_i(u,0)=Y_i(F(u))$ is tangent to $S$.  Thus every $Y_i$, and
hence every iterated bracket, is tangent to $S$; bracket generation forces
$r=\dim M$, so $S$ is open.

By density, some $z\in\mathcal A^+(x)$ lies in the interior of
$\mathcal A^-(y)$.  A positive word takes $x$ to $z$, and reversing a
negative word from $y$ to $z$ gives a positive word from $z$ to $y$.
Applying this to $-Y_1,\ldots,-Y_n$ gives the negative-time statement.
\end{proof}

\begin{samepage}
\begin{theorem}[One-sided shape control]\label{thm:one-sided}
Let $c$ be a smooth affine-equivariant center, and let $\Omega$ be a
connected open subset of its domain, invariant under translations and
positive scalings, on which $\mathcal R_c$ is symmetric and definite.  Fix a
time direction.  For all $P,Q\in\Omega$, finitely many center moves in that
direction, each contained in $\Omega$, take $P$ to $aQ+b\one^{\mathsf T}$
for some $a>0$ and $b\in\R^d$.
\end{theorem}
\end{samepage}

\begin{proof}
Factor out translation and scale by
\[
 r(P)=\|P-m(P)\one^{\mathsf T}\|_F,\qquad
 \pi(P)=\frac{P-m(P)\one^{\mathsf T}}{r(P)}.
\]
Here $r>0$ on $\Omega$, since at a collapsed configuration equivariance
forces $c(P)=p_1$ and the response vanishes.  Thus $S=\pi(\Omega)$ is a
connected open subset of the unit sphere of centered configurations.  As $c$
commutes with translations and positive scalings, $\pi$ carries each $\X_i$
to a field $Y_i$ on $S$; with $v_i=z_i-c(Z)$ for $Z\in S$,
\begin{equation}\label{eq:projected}
 Y_i(Z)=\iota_i(v_i)-\frac1n v_i\one^{\mathsf T}-(z_i^{\mathsf T}v_i)Z.
\end{equation}
These fields are bracket generating by Theorem~\ref{thm:response}, because
$\pi$ is a submersion.  Moving all vertices at once dilates the configuration
about its center, and $\pi$ removes dilations, so
\begin{equation}\label{eq:balance}
 \sum_i\X_i(P)=P-c(P)\one^{\mathsf T},\qquad \sum_iY_i=0.
\end{equation}
Lemma~\ref{lem:zero-sum} therefore gives a finite one-sided word from
$\pi(P)$ to $\pi(Q)$.

To lift this word, let $\varepsilon\in\{1,-1\}$ be the chosen direction and
write $R=rZ+b\one^{\mathsf T}$ with initial values $r=r(P)$ and $b=m(P)$.
Along each positive-time segment of $\varepsilon Y_i$, solve
\begin{equation}\label{eq:lift}
 \dot r=\varepsilon r\,z_i^{\mathsf T}v_i,\qquad
 \dot b=\frac{\varepsilon r}{n}v_i.
\end{equation}
By~\eqref{eq:projected} and~\eqref{eq:lift}, $\dot R=\varepsilon\X_i(R)$, so
the lift is a center move of vertex $i$ in the chosen direction.  The
coefficients in~\eqref{eq:lift} are bounded on the compact image of the word,
so $r$ stays positive and $r,b$ stay finite, and the invariance of $\Omega$
keeps $R$ in $\Omega$.  At the end $\pi(R)=\pi(Q)$, so $R$ is a translated,
positively scaled copy of $Q$.
\end{proof}

For the uniform centroid, the second moment
$E(P)=\sum_i\|p_i-m(P)\|^2$ constrains the attainable scale.  Since
$DE_P[\X_i]=2\|p_i-m(P)\|^2$, nonconstant inward executions strictly
decrease $E$ and nonconstant outward executions strictly increase it.  As
$E(aQ+b\one^{\mathsf T})=a^2E(Q)$, every nonconstant one-sided execution from
$P$ to $aQ+b\one^{\mathsf T}$ satisfies
\[
 a<\sqrt{E(P)/E(Q)}\quad\text{if inward},\qquad
 a>\sqrt{E(P)/E(Q)}\quad\text{if outward}.
\]
Figure~\ref{fig:one-sided} shows two exact one-sided executions, with
rational coordinates and convexity certificates in
Appendix~\ref{app:examples}.

\begin{figure}[!htbp]
\centering
\includegraphics[width=\linewidth]{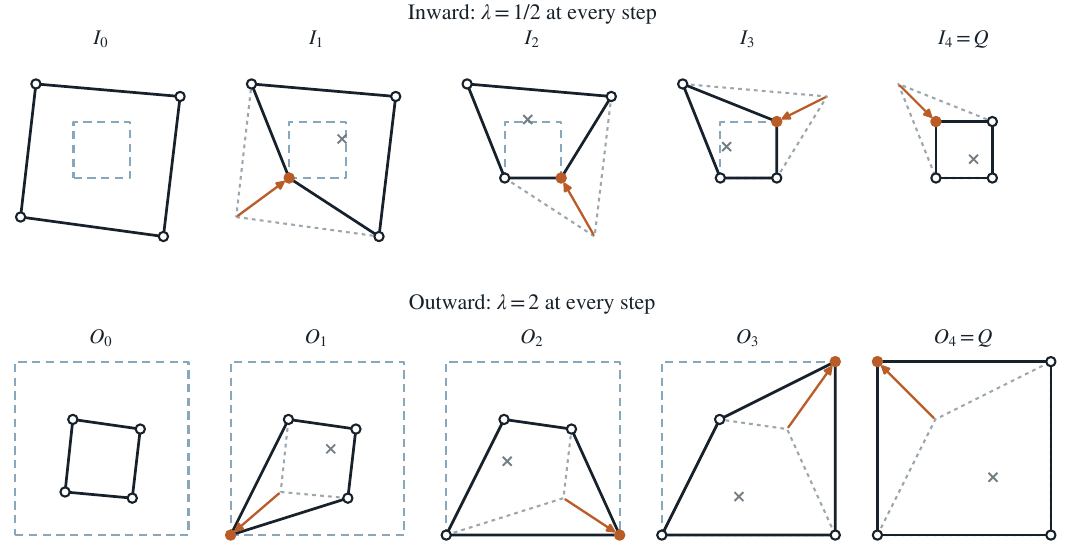}
\caption{Exact one-sided regularization of two quadrilaterals.  Each row
moves vertices $1,2,3,4$ in turn, with $\lambda=1/2$ (inward, top) or
$\lambda=2$ (outward, bottom), and ends exactly at the square $Q$ (dashed).
Orange marks the active vertex and its move, dotted gray lines show the
preceding polygon, and crosses mark the centroid of the other three vertices.
Both executions remain strictly convex.}
\label{fig:one-sided}
\end{figure}

\section{Reachability classes}\label{sec:topology}

The main theorems now follow by identifying the relevant connected open sets.

\begin{lemma}[Full-rank configuration space]\label{lem:configuration-components}
The space $\cM_{d,n}$ is path connected for $n\ge d+2$.  For $n=d+1$ it has
two path components, distinguished by the simplex orientation
$\operatorname{sgn}\det(p_2-p_1,\ldots,p_{d+1}-p_1)$.
\end{lemma}

\begin{proof}
Splitting off the centroid and writing the restriction of $P$ to
$W=\ker\one^{\mathsf T}$ in an orthonormal basis identifies $\cM_{d,n}$ with
$\R^d$ times the space of $d\times(n-1)$ matrices $B$ of full row rank.  The
polar decomposition $B=(BB^{\mathsf T})^{1/2}U$ splits $B$ into a
positive-definite factor, which ranges over a convex set, and a matrix $U$
with orthonormal rows.  If $n-1>d$, every such $U$ extends to an oriented
orthonormal basis, so the connected group $\mathrm{SO}(n-1)$ acts
transitively on these frames.  If $n-1=d$, then $U$ is orthogonal, and the
sign of $\det U$ agrees, up to a fixed sign depending on the basis, with the
simplex orientation.
\end{proof}

\begin{samepage}
\begin{theorem}[Full-dimensional configurations]\label{thm:configurations}
Positive vertex-centroid moves join any two configurations in a connected
open set $\Omega\subseteq\cM_{d,n}$ by an execution that stays in $\Omega$.
In particular, they join any two affinely spanning configurations when
$n\ge d+2$, while for $n=d+1$ the reachability classes are the two simplex
orientations.  Within each class, inward-only moves, and likewise
outward-only moves, take any configuration to a translated, positively
scaled copy of any other.
\end{theorem}
\end{samepage}

\begin{proof}
The first assertion follows from Theorem~\ref{thm:response}
by~\eqref{eq:discrete-covariance} and~\eqref{eq:arithmetic-flow}.
Lemma~\ref{lem:configuration-components} supplies the classes, and positive
moves cannot leave them: they multiply $\widehat P$ on the right by matrices
of positive determinant (Appendix~\ref{app:words}), preserving affine rank
and simplex orientation.  The classes are invariant under translations and
positive scalings, so Theorem~\ref{thm:one-sided} gives the one-sided
assertion.
\end{proof}

\begin{lemma}[Simple-polygon space]\label{lem:polygon-components}
For each $\sigma\in\{1,-1\}$, the simple labeled $n$-gons of orientation
$\sigma$ form a connected open set $\cS_n^\sigma\subset\cM_{2,n}$.
\end{lemma}

\begin{proof}
\emph{Openness.}  In a simple polygon, nonadjacent edges are at positive
distance, as are vertices and nonincident edges, and these separations
persist under small perturbations.  At a vertex $p_i$, with
$u=p_{i-1}-p_i$ and $v=p_{i+1}-p_i$, the two incident edges meet only at
$p_i$ exactly when
\[
 u,v\ne0,\qquad \det(u,v)\ne0\ \text{ or }\ u^{\mathsf T}v<0,
\]
an open condition that includes straight angles.  Hence simplicity is open,
and so is each orientation, since the signed area is continuous and nonzero
on simple polygons.

\emph{Connectivity.}  We show by induction on $n$ that every simple $P$ can
be joined to every strictly convex $Q$ of the same orientation.  For
triangles, an orientation-preserving affine path suffices.  Let $n\ge4$.
After a small perturbation within $\cS_n^\sigma$, no three vertices of $P$
are collinear.  A leaf triangle of any triangulation of $P$ is then a clean
ear~\cite{Meisters1975}: the closed triangle meets the boundary arc that avoids
the ear vertex only at the diagonal endpoints.  Moving the ear vertex linearly to the
midpoint of the diagonal keeps the polygon simple and ends at a straight
angle.  Deleting this vertex leaves a simple $(n-1)$-gon $P^-$ with the signed
area of the flattened polygon and hence the orientation of $P$.

Delete the vertex with the same label from $Q$ to obtain a strictly convex
$Q^-$, and join $P^-$ to $Q^-$ through simple polygons by induction.
Reinserting the midpoint of the corresponding edge along this path gives a
path of simple $n$-gons with the original labels.  Finally, move the midpoint
out to the deleted vertex of $Q$; by strict convexity, the ear triangle lies
on the opposite side of its diagonal from the rest of $Q$, so the polygon
stays simple.  Thus every polygon in $\cS_n^\sigma$ is joined to a fixed
regular polygon, and $\cS_n^\sigma$ is connected.
\end{proof}

\begin{proof}[Proof of Theorem~\ref{thm:polygons}]
By~\eqref{eq:discrete-covariance} and
Proposition~\ref{prop:volume-response}, both centers have positive-definite
response on $\cS_n^\sigma$, a connected open set invariant under translations
and positive scalings.  Theorems~\ref{thm:response} and~\ref{thm:one-sided},
together with~\eqref{eq:arithmetic-flow} for the vertex centroid, show that
(i) implies (ii) and (iii).  Conversely, the signed area is continuous and
nonzero along an execution through simple polygons, so its sign is constant;
translations and positive scalings preserve it as well.
\end{proof}

\begin{proof}[Proof of Corollary~\ref{cor:topp}]
A noncollinear input lies in $\cM_{2,n}$, so Theorem~\ref{thm:configurations}
takes it to a regular polygon, of matching orientation when $n=3$, by
inward-only or by outward-only moves.  Under the immobile convention, every
move, whatever the sign of $\lambda$, keeps all vertices in their current
affine hull, so a collinear input cannot become regular.

Under the permissive convention, suppose first that the input is collinear
but not a single point.  Since $n\ge3$, we can choose $i$ so that the other
vertices are not all equal.  After a move with $\lambda=0$ if necessary,
$p_i=c_i(P)$ is the centroid, and a permissive departure off the line makes
the configuration noncollinear, so the first case applies.  If all vertices
coincide, one permissive departure first creates two distinct points.
\end{proof}

In Corollary~\ref{cor:topp}, $\cM_{2,n}$ is connected for $n\ge4$, so a
noncollinear input can be taken exactly to any prescribed regular polygon;
for nondegenerate triangles, positive moves
preserve orientation and moves with $\lambda<0$ reverse it.

For area and volume centroids, Theorems~\ref{thm:response}
and~\ref{thm:one-sided} apply to each connected component of the realization
domains of Section~\ref{sec:covariance}.
Each such component is invariant under translations and positive scalings,
since these transformations are connected to the identity.

\section{Finite words and move counts}\label{sec:counts}

For the arithmetic centroid, matrix representations of the moves yield a
local word-length bound, with inward and outward moves allowed in the same
execution.

\begin{proposition}[Local word length]\label{prop:local-length}
For every $P\in\cM_{d,n}$ and every open neighborhood $U$ of $P$, there is a
neighborhood $V$ of $P$ such that every $Q\in V$ is reachable from $P$ in at
most $3dn$ positive vertex-centroid moves, with the full execution in $U$.
\end{proposition}

Appendix~\ref{app:words} builds $dn$ independent endpoint directions, each
from at most three moves, and applies the inverse function theorem.
The same computation shows that
positive moves generate exactly the invertible matrices with unit column
sums and positive determinant (Theorem~\ref{thm:groups}).

\begin{proposition}[Generic dimension bounds]\label{prop:lower-bounds}
Consider finitely many move types on an $N$-dimensional manifold $M$, each a
smooth map from an open subset of $M\times\R$ to $M$.  From a
fixed source, the targets reachable in fewer than $N$ moves form a set of
measure zero.  If, in addition, each move with fixed parameter is a local
diffeomorphism and $\mathcal N$ is an embedded $r$-dimensional submanifold,
then the sources reaching $\mathcal N$ in fewer than $N-r$ moves form a set
of measure zero.
\end{proposition}

\begin{proof}
For a word of length $k$, the executable source--parameter pairs form an open
subset of $M\times\R^k$.  With the source fixed, the endpoint map has a
$k$-dimensional domain, so Sard's theorem gives the first assertion.  The
joint endpoint map of source and parameters is a submersion, because its
derivative in the source is invertible, so the preimage of $\mathcal N$ has
dimension $k+r$.  When $k+r<N$, its projection to the source manifold has
measure zero, again by Sard's theorem.  Only finitely many sequences of move
types lie below either threshold.
\end{proof}

Positive vertex-centroid moves and center flows satisfy these hypotheses
with $N=dn$.  For affinely spanning configurations in $\R^d$, generic
targets thus need at least $dn$ moves from a fixed source, and generic
sources need at least $dn-d-1$ moves to reach a prescribed shape, whose
copies form a $(d+1)$-dimensional family.  Regular $n$-gons of a given
orientation form a four-dimensional family (translation, rotation, scale),
so generic regularization needs at least $2n-4$ moves.

\section{Discussion}\label{sec:discussion}

On an open domain with symmetric definite center response, reversible
reachability is determined by the connected components.  When these
components are invariant under translations and positive scalings, the
zero-sum identity~\eqref{eq:balance} also gives one-sided shape control.

The main quantitative problem is to turn the local words into an effective
global algorithm with move counts depending on geometric clearance and
conditioning.  Generic regularization
needs at least $2n-4$ positive moves, and Hawley's permissive construction
uses $2n-4$ moves when no remedies for degenerate placements are
needed~\cite{Hawley2014}; the sharp count under the immobile convention is
open.  Further questions concern the translations and scales attainable by
one-sided control and the
connectivity of higher-dimensional polyhedral realization spaces.

\section*{Declaration of generative AI assistance}
ChatGPT assisted with language editing and parts of the mathematical
development, including the area-centroid analysis,
the general framework for affine-equivariant centers, and the one-sided shape
control argument.  The author set the research direction and finalized the
arguments and presentation.

\clearpage
\appendix
\section{Explicit words and the move group}\label{app:words}

Let $e_1,\ldots,e_n$ be the standard basis of $\R^n$, and set
\[
 w_i=e_i-\frac1{n-1}\sum_{j\ne i}e_j,\quad
 X_i=w_ie_i^{\mathsf T},\quad
 A_i(\lambda)=I+(\lambda-1)X_i.
\]
Then $Pw_i=p_i-c_i(P)$, so right multiplication $P\mapsto PA_i(\lambda)$ is
exactly the move~\eqref{eq:move}.  Since $\one^{\mathsf T}w_i=0$ and
$e_i^{\mathsf T}w_i=1$,
\[
 \one^{\mathsf T}A_i=\one^{\mathsf T},\quad X_i^2=X_i,\quad
 A_i(\lambda)A_i(\mu)=A_i(\lambda\mu),\quad \det A_i(\lambda)=\lambda.
\]
The move is executed through the matrices $A_i(1-s+s\lambda)$,
$0\le s\le1$, which are invertible for $\lambda>0$; hence
$\widehat{PA_i}=\widehat PA_i$ keeps the rank of $\widehat P$ throughout.  If
$Pw_i=0$, the move is constant and can be omitted.

Let $W=\ker\one^{\mathsf T}$, $\fg=\{Y\in M_n(\R):\one^{\mathsf T}Y=0\}$,
and $F_{ij}=w_ie_j^{\mathsf T}$ for all $i,j$.  The $w_i$ span $W$, so the
$F_{ij}$ span $\fg$.

\begin{lemma}[Three-factor directions]\label{lem:three-factor}
For every $a>0$ with $a\ne1$, the following curves have derivatives at
$t=0$ spanning $\fg$:
\[
 A_i(e^t),\qquad
 C_{ij}^\pm(t)=A_i(a^{\pm1})A_j(e^t)A_i(a^{\mp1})\quad(i<j).
\]
\end{lemma}

\begin{proof}
The first curves give $X_i=F_{ii}$.  Put $\alpha=-1/(n-1)$.  For $i\ne j$,
$X_iX_j=\alpha F_{ij}$ and $X_iX_jX_i=\alpha^2X_i$, so
\[
 Y_{ij}^{\pm}:=(C_{ij}^{\pm})'(0)
 =X_j+\alpha(a^{\pm1}-1)F_{ij}+\alpha(a^{\mp1}-1)F_{ji}
   +\alpha^2(a-1)(a^{-1}-1)X_i.
\]
Modulo the span of the $X_k$, their difference and sum are
\[
 Y_{ij}^+-Y_{ij}^-=\alpha(a-a^{-1})(F_{ij}-F_{ji}),\qquad
 Y_{ij}^++Y_{ij}^-\equiv\alpha(a+a^{-1}-2)(F_{ij}+F_{ji}).
\]
Both coefficients are nonzero, so every $F_{ij}$ is recovered.
\end{proof}

\begin{proof}[Proof of Proposition~\ref{prop:local-length}]
For $P\in\cM_{d,n}$, the restriction $P|_W$ maps onto $\R^d$: solve
$\widehat Py=(z,0)^{\mathsf T}$ for any $z\in\R^d$.  Column by column, this
gives $\{PY:Y\in\fg\}=M_{d,n}(\R)$.  Choose a ball $B_\rho(P)\subset U$, and
choose $a\ne1$ so close to $1$ that the executions of all three-factor words
at $t=0$, each of which multiplies to $I$, stay in $B_{\rho/2}(P)$.  By
Lemma~\ref{lem:three-factor}, some $dn$ of the curves, $C_1,\ldots,C_{dn}$,
make the matrices $PC_k'(0)$ a basis of $M_{d,n}(\R)$, so
\[
 \Phi(t_1,\ldots,t_{dn})=PC_1(t_1)\cdots C_{dn}(t_{dn})
\]
has invertible derivative at $0$.  Every partial move depends continuously on
$t$ and on the execution parameter $s\in[0,1]$, so by compactness all partial
executions stay in $B_\rho(P)$ for small $t$.  By the inverse function
theorem, these words reach a neighborhood $V$ of $P$, each with at most
$3dn$ positive moves.
\end{proof}

Lemma~\ref{lem:three-factor} is a finite-word form of the bracket relations
\[
 [X_i,X_j]=\alpha(F_{ij}-F_{ji}),\qquad
 [X_i,F_{ij}-F_{ji}]=F_{ij}+F_{ji}-2\alpha X_i,
\]
which give $\Lie(X_1,\ldots,X_n)=\fg$.  As $a\to1$, the antisymmetric and symmetric coefficients
in the lemma vanish to first and second order, respectively.

\begin{theorem}[Exact move groups]\label{thm:groups}
Let $G=\{A\in\mathrm{GL}_n(\R):\one^{\mathsf T}A=\one^{\mathsf T}\}$ and
$G^+=\{A\in G:\det A>0\}$.  For $n\ge2$,
\[
 \langle A_i(\lambda):\lambda>0\rangle=G^+,\qquad
 \langle A_i(\lambda):\lambda\ne0\rangle=G.
\]
\end{theorem}

\begin{proof}
In the decomposition $\R^n=W\oplus\R(\one/n)$, elements of $G$ have block
form $\left(\begin{smallmatrix}B&b\\0&1\end{smallmatrix}\right)$, so
$G^+\cong\mathrm{GL}^+(n-1,\R)\ltimes\R^{n-1}$ is
connected~\cite{Poole1995}.  Lemma~\ref{lem:three-factor} at $I$ provides
$n(n-1)$ curves whose derivatives form a basis of $\fg$, so their product
endpoint map has invertible derivative.  The subgroup generated by positive
moves therefore contains a neighborhood of $I$; being open in the connected
group $G^+$, it equals $G^+$.  The factor $A_1(-1)$ supplies the other
determinant component.
\end{proof}

\section{Response examples and exact executions}\label{app:examples}

\paragraph{A direction missing from the first brackets.}
For the uniform centroid at the square $((1,1),(-1,1),(-1,-1),(1,-1))$,
$d_i=q_i/4$.  With $H=\operatorname{diag}(1,-1)$, the nonzero covector with blocks
$\ell_i=Hq_i/4$ annihilates every field $\X_i$ and every first bracket,
because $H$ is symmetric and $q_i^{\mathsf T}Hq_i=0$ for each $i$.  The four
fields $\X_i$ and $[\X_1,\X_2]$, $[\X_1,\X_4]$, $[\X_2,\X_3]$ are linearly
independent, so the rank is exactly $7$; by Theorem~\ref{thm:response},
higher brackets recover the remaining direction.

\paragraph{Symmetry and degeneracy.}
Affine equivariance does not force a symmetric response.  On the
affine-invariant open set where $\det(p_2-p_1,p_4-p_1)\ne0$, let
\[
 h(P)=\frac{\det(p_2-p_1,p_3-p_1)}{\det(p_2-p_1,p_4-p_1)},\qquad
 c(P)=m(P)+h(P)(p_1-p_2).
\]
At $P=((0,0),(1,0),(2,1),(0,2))$, direct differentiation gives
\[
 (d_1\ d_2\ d_3\ d_4)=\frac1{16}
 \begin{pmatrix}6&-15&5&4\\-9&3&1&5\end{pmatrix},\qquad
 \mathcal R_c(P)=\frac1{16}\begin{pmatrix}-5&13\\5&11\end{pmatrix},
\]
yet the fields $\X_i$ and their first brackets have full rank $8$ there.  By
contrast, the intersection of the diagonal lines $p_1p_3$ and $p_2p_4$,
defined where they are transverse, has $d_i=0$ for every $i$, because a
radial motion keeps each vertex on its diagonal line.  The fields commute,
and their nonzero values lie in distinct vertex-coordinate blocks.  Their
rank, and hence the orbit dimension, equals the number of vertices distinct
from the center.

\paragraph{Exact one-sided executions.}
Let
\[
 Q=\begin{pmatrix}-1&1&1&-1\\-1&-1&1&1\end{pmatrix},\qquad
 \Pi(\lambda)=A_1(\lambda)A_2(\lambda)A_3(\lambda)A_4(\lambda).
\]
The four-move examples in Figure~\ref{fig:one-sided} use the explicitly
constructed sources
\begin{align*}
 P_{\mathrm{in}}=Q\Pi(1/2)^{-1}
  &=\frac1{81}\begin{pmatrix}-233&177&225&-189\\-193&-249&153&189\end{pmatrix},\\
 P_{\mathrm{out}}=Q\Pi(2)^{-1}
  &=\frac1{324}\begin{pmatrix}-137&114&144&-108\\-163&-186&72&108\end{pmatrix}.
\end{align*}
Over the five endpoint polygons $R=(r_1,\ldots,r_4)$, all directed edges
$r_jr_{j+1}$, and all remaining vertices $r_k$, the minimum of
$\det(r_{j+1}-r_j,r_k-r_j)$ is $4$ in the inward example and $101/162$ in the
outward example.  Each determinant is affine during a move, so endpoint
positivity certifies strict convexity throughout.

\end{document}